\documentclass[journal,draftclsnofoot,12pt,onecolumn]{IEEEtran}
\ifCLASSINFOpdf
\else
\fi
\usepackage{amsmath}
\usepackage{amssymb}
\usepackage[linesnumbered,ruled]{algorithm2e}
\usepackage{todonotes}
\usepackage{times,amsxtra,latexsym,graphics,verbatim,epsfig,setspace,subfigure,epic}
\DeclareMathOperator*{\inv}{In}

\newtheorem{definition}{Definition}[section]
\newtheorem{theorem}{Theorem}[section]
\newtheorem{remark}[theorem]{Remark}

\newtheorem{lemma}[theorem]{Lemma}

\begin{document}
%
\title{Synchronizing Rankings via Interactive Communication}
%
%
%

\author{Lili~Su,~\IEEEmembership{Student Member,~IEEE,}
        and~Olgica~Milenkovic,~\IEEEmembership{Senior Member,~IEEE}

\thanks{Lili Su and Olgica Milenkovic are both with the Department
of Electrical and Computer Engineering, University of Illinois at Urbana--Champaign, Urbana,
IL, 61801 USA e-mail: \{lilisu3,milenkov\}@illinois.edu.
}
}

\maketitle

\begin{abstract}
We consider the problem of exact synchronization of two rankings at remote locations connected by a two-way channel. Such synchronization problems arise when items in the data are distinguishable, as is the case for playlists, tasklists, crowdvotes and recommender systems rankings. Our model accounts for different constraints on the communication throughput of the forward and feedback links, resulting in different anchoring, syndrome and checksum computation strategies. Information editing is assumed of the form of deletions, insertions, block deletions/insertions, translocations and transpositions. The protocols developed under the given model are order-optimal with respect to genie aided lower bounds.
\end{abstract}


%
\IEEEpeerreviewmaketitle

\section{Introduction}
%
%
%
%
Rankings are emerging data formats that capture information about orderings of elements, and they include linear orders, weak orders -- orders with ties, and partial orders. Linear orders are most frequently referred to as permutations, as they involve distinct elements, while weak orders are sometimes known as multiset permutations. Ranking formats appear in a wide variety of applications, including social choice theory, where one is concerned with ranking candidates based on their suitability for a certain position~\cite{sen1986social}, search and meta-search engines, where one is concerned with ranking web-pages according to their relevance with respect to search keywords~\cite{joachims2002optimizing}, and bioinformatics and gene prioritization, where one ranks genes according to their likelihood of being involved in a disease, or where one is concerned with rearrangements of unique genetic blocks within different genomes~\cite{aerts2006gene}. In addition, permutations have found applications for efficient encoding of automata and sequences~\cite{kari1996representation}, while both permutations and multiset permutations are frequently used for encoding binary relations between objects. Many popular voting sites store large volumes of ordinal and relational data, frequently based on pairwise comparisons, and examples include CrowdVoting systems such as Reddit, Heycrowd, and KittenWar~\cite{lewry2007kittenwar}. Permutations are reconstructed based on a sufficiently large number of informative pairwise comparisons, which are in one-to-one correspondence with binary relations~\cite{barbay2007adaptive}.

A number of ordinal data processing systems call for synchronization of their ranking information at remote locations, within static or dynamically changing data acquisition environments. Here, synchronization refers to reaching a consensus ranking or reconstructing a ranking at one node based on partial information given at another node of the network. Different nodes may contain different versions of a file containing ordinal data, such as for example data reflecting preference orders for movies, politicians, food choices, music playlists and other items.

Other important examples pertain to distributed and metasearch engine systems, where information about millions of dynamically changing web-pages is stored, and routing engines, storing large volumes of priority information. In the former case, of particular interest are rankings of web-pages which have to be constructed using some sorting criteria or algorithm, such as PageRank, specific to data at a given location. For example, at one location, one may have full access to the web-pages and their scores, while at another, only a partial order may be available, reflecting the scores of a reduced number of web-pages. Every time a web-page is updated, the score of the web-page changes as well. This change in score may consequently change the ranking of the web-pages. Running PageRank is a complex, time- and energy-consuming operation and it may be desirable to quickly estimate the similarity of rankings~\cite{sun2007communication,gopalan2007estimating} between different engines and synchronize their content if required. Other emerging distributed storage systems in which synchronization of permutations may be required includes flash memories in the cloud~\cite{feuerlicht2013can}, due to the fact that rank modulation coding represents a desirable and efficient means of information storage in flash memories.

Synchronization of binary and non-binary data through interactive communication was first described in~\cite{orlitsky1987communication,orlitsky2001practical}, and extended to synchronization of sets and related entities in~\cite{orlitsky2001practical,minsky2003set,VenkataramananSR13,bitouze2013synchronization,6325200}. A number of synchronization protocols are implemented in practice, such as rsync and dsync~\cite{knauth2013dsync}, and used in dropbox and other file reconciliation systems. Nevertheless, no results on efficient synchronization protocols for permutations are currently known.

The problem we consider in this context may be succinctly stated as follows: A transmitter and a receiver, connected by a two-way noiseless channel, are placed at different locations. Each link has a total communication throughput (i.e., the largest number of bits communicated through the link within a synchronization procedure), which for the forward and feedback links equal $c_{tr}$ and $c_{rt}$, respectively.
The transmitter stores ordinal information of the form of a (partial) permutation $\sigma^{X}$, while the receiver stores a ``noisy'' version of $\sigma^{X}$, denoted by $\sigma^{Y}$.
Ordinal data noise refers to random deletions/insertions, block deletions/insertions, translocation and transposition errors. The problem of interest is to \emph{exactly} restore $\sigma^{X}$ at the receiver with the smallest two-way communication throughput between the transmitter and the receiver. In general, this problem is difficult; we therefore focus on two simplified models:
\begin{itemize}
\item \emph{The classical model:} In this case, $c_{tr} \simeq c_{rt}$, i.e., the communication throughputs of the forward and feedback links are of the same order. This case represents a generalization of the binary data scenario addressed in~\cite{VenkataramananSR13, venkataramanan2010interactive}, to ordinal information.
\item \emph{The limited feedback model:} In this case, we assume that $c_{tr}\gg c_{rt}$, or more precisely, that $c_{tr}=O(d \log\, n)$, and $c_{rt}=O(d \log \, d)$, where $n$ is the length of the ordinal message, while $d$ is the number of editing errors. Using the feedback link is costly, and for this channel, synchronization has to be achieved with a number of bit transmissions proportional to $d \, \log d$, but independent on the length of the message $n$.
\end{itemize}
Our main contributions are as follows. For $\sigma^{Y}$ and $\sigma^{X}$ mis-synchronized by deletions, we exhibit protocols within a factor of two and a factor of five from the genie-aided limits for $c_{tr} \simeq c_{rt}$ and $c_{tr}\gg c_{rt}$, respectively. When the synchronization error is a single translocation, a protocol within a factor of three from the genie-aided limit is proposed. For single transposition errors, we describe a one-way protocol within a factor of six from the genie-aided limit. This protocol uses generalization of Varshamov-Tenengolz and Reed-Solomon codes for ordinal information.


The paper is organized as follows. Section~\ref{sec:notation} contains the mathematical preliminaries and the problem formulation. Synchronization from deletions or insertions is analyzed in Section  \ref{sec: deletions}. A discussion of translocation and transposition error synchronization methods is presented in Section \ref{sec: translocation} and Section \ref{sec: transposition}, respectively.
\section{Notation and Preliminaries} \label{sec:notation}
A permutation $\sigma: [n] \to [n]$ is a bijection over $[n]\triangleq \{1, \cdots, n\}$. The collection of all permutations on $[n]$ is denoted by $\mathbb{S}_n$.
For any $\sigma \in \mathbb{S}_n$, we write $\sigma = (\sigma_1, \sigma_2,\cdots, \sigma_n)$, where $\sigma_i$ is the image of $i \in [n]$ under $\sigma$. The identity permutation $(1,2,\cdots, n)$ is denoted by $e$.

The projection of a permutation $\sigma$ onto a set $P \subseteq [n]$, denoted by $\sigma_P$, is obtained by removing all elements in $[n]\setminus P$ from $\sigma = (\sigma_1, \sigma_2,\cdots, \sigma_n)$. In particular, when $P=[n]$, $\sigma_P=\sigma$.
As an example, $(2,3,7,5,1)$ is the projection of a permutation over any $[n]$ for which $n \geq 7$ onto the set $P=\{1, 2, 3, 5, 7\}$. We tacitly assume that $n$ is either known in advance, or that it equals to the value of the largest element in the partial permutation. Which of these assumptions is used will be apparent from the context.
We frequently refer to projections as partial permutations and do not explicitly write the subscript $P$ unless required by the context.

Given $\sigma_P$, a deletion refers to removing an element in $P$ from $\sigma_P$. Similarly, an insertion refers to inserting an element in $[n]\setminus P$ into an arbitrary position of $\sigma_P$. A block of deletions or insertions of length $d$ corresponds to a set of deletions or insertions contained within $d$ consecutive positions.
A swap of two elements in a permutation is referred to as a \emph{transposition}. For example, the symbols $1$ and $2$ are transposed in $(2,1,3,4)$ when compared to the identity permutation $(1,2,3,4)$.
A pair of an insertion and a deletion involving the same element is termed a translocation~\cite{FarnoudSM13}, formally defined next.
\begin{definition}
A \emph{translocation} $\varphi{(i,j)}$ is a permutation defined as follows: If $i \le j$, we have \[\varphi{(i,j)}=(1, \cdots, i-1, i+1, \cdots,j-1, j, i, j+1, \cdots, n),\]
and if $i > j$, we have \[ \varphi{(i,j)}= (1, \cdots, j-1, i, j, j+1, \cdots, i-1, i+1, \cdots, n) \; . \] For $i\le j$, the permutation $\varphi{(i,j)}$ is called a \emph{right} translocation while the permutation $\varphi{(j,i)}$ is called a \emph{left} translocation. Translocations arise due to independent \emph{falls} and \emph{rises} of elements in a ranking.
\end{definition}
\begin{definition}
The inversion vector of $\sigma_P$, denoted by $\inv(\sigma_P)$, is a binary vector $(x_1,\cdots,x_{|P|-1})$, such that
\[ x_i=
\begin{cases}
1, & \text{if}~ \,\sigma_i>\sigma_{i+1};\\
0, & \text{if}~ \,\sigma_i<\sigma_{i+1}.
\end{cases}
\]
\end{definition}
In our subsequent analysis, we also make use  of Varshamov-Tenengolz codes VT$_a(n)\subseteq \{0,1\}^n$. These codes consist of all binary vectors $(x_1, \cdots, x_n)$ satisfying the congruence
\begin{equation}
\sum_{i=1}^{n} \, i\cdot x_i \equiv a \mod (n+1),
\end{equation}
where the parameter $a\in \{0,1,\cdots, n\}$ is referred to as the VT-syndrome of the code VT$_a(n)$.
VT-codes are single deletion error-correcting codes, which is easily proved by exhibiting a decoding algorithm~\cite{levenshtein1966bcc,Levenshtein65}.

The family of VT-codes partitions the space $\{0,1\}^{n}$ into $n+1$ single deletion correcting codes~\cite{Lev91}. A less known result holds for permutations, asserting that $\mathbb{S}_n$ may be partitioned into $n$ cosets of size $(n-1)!$, each of which has a unique VT-syndrome for all the inversion vectors. The cosets represent single deletion correcting codes for permutations. The key observations behind the proof of this fact are that: a) a single deletion in the permutation induces a single deletion in the inversion vector; b) a deletion in the inversion vector may be corrected via VT coding; and c) given a letter $b$ in $[n]\setminus P$ and a binary string $\mathcal{B}$ which produces the inversion vector $\inv(\sigma_P)$ via a single deletion, there is a unique way to insert $b$ into $\sigma_P$ such that the newly obtained partial permutation has inversion vector $\mathcal{B}$. 

Throughout the paper, we assume that $n$ and the number of deletion (insertion) errors $d$ is known in advance both to the transmitter and receiver; that all $\binom{n}{d}$ deletion (insertion) patterns are equally likely; and that the transmitter and receiver can agree in advance on the steps of the synchronization protocol. For the case of block errors, we also assume that the span of the block $d$ is known both to the transmitter and receiver; and that all $d$-spans are equally likely. Due to the complicated nature of translocation and transposition errors, we focus only on single error events and relegate the generalization to multiple errors to the journal version of the paper. Although there is no fundamental limitation in allowing $d=O(n)$, for simplicity of exposition, we restrict our attention to the case $d=o(n)$.

\section{Synchronization from Deletions/Insertions} \label{sec: deletions}

The first problem we address is synchronization from deletion errors only. In this case, $\sigma^{Y}$ is generated from $\sigma^{X}$ by deleting $d$ symbols.
\subsection{Synchronization from random deletions/insertions}\label{sec: randomdeletions}
Assume that $\sigma^{X}\in \mathbb{S}_n$ and that the transmitter is aided by a genie that knows the locations of the deleted symbols in the receiver's partial permutation $\sigma^{Y}$. Since there are $\binom{n}{d}$ possible positions for the $d$ deleted symbols and $d!$ possible orderings of the deleted symbols, the transmitter needs to send
\[
\log\, \binom{n}{d}d!=d(\, \log\, n +o(1))
\] bits, in order to enable the receiver to reconstruct $\sigma^{X}$. 

The solution in the classical setting is straightforward, described  in Protocol 1. The key observation is that the receiver can deduce the identity of the missing symbols, given that he knows $n$. Hence, the receiver sends $\log \binom{n}{d}$ bits to the transmitter indicating the missing symbols, and the transmitter in return sends the locations of the missing symbols along with their ordering. In this way, $\sigma^{X}$ can be reconstructed at the receiver with a total number of
\[
\log\,\binom{n}{d}+\log\, \binom{n}{d}d!= d(2\log\,n-\log\, d + O(1))
\]
transmitted bits, which is only twice as much as required by a genie-aided method. However, this approach cannot be used in the limited feedback scenario,
given that the throughput of the feedback link is not allowed to scale as $d\, \log \, n$.%
\begin{algorithm}
\caption{Identical Throughput Protocol}
\BlankLine
The receiver sends the identities of the $d$ deleted symbols\;
Transmitter $T$ sends the locations of the $d$ deleted symbols as well as their ordering.
\label{naive}
\end{algorithm}

We next propose a protocol for the limited feedback scenario that is within a factor of five from the genie-aided result.

As part of the protocol, the transmitter maintains a list $L_{\sigma^{X}}$, whose entries consist of the unsynchronized substrings of $\sigma^{X}$. This list is initialized to $L_{\sigma^{X}}=\{\sigma^{X}\}$. Similarly, the receiver maintains a corresponding list of unsynchronized substrings, denoted by$L_{\sigma^{Y}}$, initialized to $L_{\sigma^{Y}}=\{\sigma^{Y}\}$. The limited feedback protocol is described in Protocol 2.

\begin{algorithm}
\caption{Limited Feedback Protocol}
\BlankLine
Initialization: $L_{\sigma^{X}}\gets\{\sigma^{X}\}$, $L_{\sigma^{Y}}\gets\{\sigma^{Y}\}$, $i\gets0$\;
\While{$L_{\sigma^{X}}\not=\varnothing$ and $d>1$}
{\For{$i=1:1:|L_{\sigma^{X}}|$}
{Receiver requests the transmitter to send the central symbol of $L_{\sigma^{X}}(i)$\;
\eIf{Receiver cannot find a match for the central symbol}
{$d\gets d-1$\;}
{\eIf{the central symbol was not shifted to the left}
{There is no deletions in the left half of substring $L_{\sigma^{X}}(i)$}
{\eIf{the central symbol was shifted to the left by one}
{Receiver requests the VT-syndrome and the checksum $\Sigma$ of the left half of  substring $L_{\sigma^{X}}(i)$ and sets $d \gets d-1$\;}
{Receiver adds the left half of substrings $L_{\sigma^{X}}(i)$ and $L_{\sigma^{Y}}(i)$ to the lists $L_{\sigma^{X}}$ and $L_{\sigma^{Y}}$, respectively\;}
}
Repeat step 8--step 16 for the right half of substring $L_{\sigma^{X}}(i)$\;
}
Transmitter and receiver remove $L_{\sigma^{X}}(i)$ and $L_{\sigma^{Y}}(i)$ from $L_{\sigma^{X}}$ and $L_{\sigma^{Y}}$, respectively\;
}
}
\end{algorithm}
The idea of the protocol is to first partition $\sigma^{X}$ into a set of substrings each of which contains one deleted symbol, akin to~\cite{venkataramanan2010interactive}. Partitioning is achieved via a sequence of transmissions of a \emph{single anchor symbol}, positioned in the middle of substrings of interest. To correct a single deletion error within each substring, the receiver needs to know both the deleted symbol in that substring and the deleted position, which can be deduced from the \emph{checksum} and the VT-syndrome of the inversion vector of the substring, respectively. Here, the checksum of a substring refers to the sum of its corresponding symbols. The identity of the deleted symbol in a specified substring can be found by computing the difference of the checksum of the substring in $\sigma^{X}$ and the checksum of the corresponding noisy substring in $\sigma^{Y}$. Once the identities of the deleted symbols within the substrings are known to the receiver, synchronization is accomplished via VT coding.

Two observations are in place. Given that the data consists of distinct symbols, erroneous matching is not possible. The most costly steps of synchronization are checksum transmissions, all of which take place over the forward channel.

\begin{theorem}
Protocol 2 exactly restores $\sigma^{X}$ at the receiver, with
\begin{equation}
\mathbb{E}[N_{T\to R}(d)]\le (5d-2)\log n-2d\log d-d\log 2, \notag
\end{equation}
and
\begin{equation}
\mathbb{E}[N_{R\to T}(d)]\le 6(d-1). \notag
\end{equation}
\label{thm1}
\end{theorem}
\vspace{-0.05in}
\begin{proof}
The protocol provides an exact solution, since one cannot make errors in the process of anchoring the central symbol.

When synchronizing from $d$ deletions, the total number of bits transmitted from the transmitter to the receiver until Protocol 2 terminates may be written as
\begin{equation}
N_{T\to R}(d)=N_{c}(d)+N_{v}(d)+N_{s}(d),
\end{equation}
where $N_{c}, N_{v}$ and $N_{s}$ represent the number of bits sent for the central anchor symbols, bits for the VT-syndrome of the inversion vector and bits for the checksums, respectively.

First, we show by induction that  for $d\ge 1$,
\begin{equation}
\mathbb{E}[N_{c}(d)] \le 2(d-1) \log n,
\label{induction1}
\end{equation}
where $N_{c}(0)=0$ by definition. Note that $N_c(d)$ depend both on the number of deletions and the length of the partial permutation. In our analysis, we write the dependence on $n$ explicitly as $N_c(d, n)$. In addition,
we observe that $N_c(d, n)$ is increasing in $n$.

Base Case:  $N_{c}(0,n)=N_{c}(1,n)=0$, and thus, $(\ref{induction1})$ holds.

Induction Hypothesis: Suppose that $\mathbb{E}[N_{c}(k,n)] \le 2(k-1)\log n, ~ \forall\, k\le d-1$.

Induction Step:  $\mathbb{E}[N_{c}(d,n)]$ can be rewritten by conditioning on the outcome of the first round of the protocol as:
\begin{align}
\mathbb{E}[N_{c}(d,n)]&=\log n+\frac{\binom{n-1}{d-1}}{\binom{n}{d}} \mathbb{E}[N_{c}(d-1,n-1)]\nonumber\\
&\quad+\frac{\binom{n-1}{d}}{\binom{n}{d}}\sum_{j=0}^{d}\frac{1}{2^d}\binom{d}{j} \big{(}\mathbb{E}[N_{c}(j,\lfloor \frac{n+1}{2}\rfloor)]+\mathbb{E}[N_{c}(d-j,\lfloor \frac{n+1}{2}\rfloor)]\big{)}\\
&\le \log n+\frac{d}{n} \mathbb{E}[N_{c}(d-1,n)]\nonumber\\
&\quad+\frac{\binom{n-1}{d}}{\binom{n}{d}}\sum_{j=0}^{d}\frac{1}{2^d}\binom{d}{j} \big{(}\mathbb{E}[N_{c}(j,n)]+\mathbb{E}[N_{c}(d-j,n)]\big{)}\\
&=\log n+\frac{d}{n}\mathbb{E}[N_c(d-1,n)]\nonumber\\
&\quad+\frac{n-d}{n2^d}\Big{(}2\mathbb{E}[N_c(d,n)]+\sum_{j=1}^{d-1} \binom{d}{j} \big{(}\mathbb{E}[N_{c}(j,n)]+\mathbb{E}[N_{c}(d-j,n)]\big{)}\Big{)}
\end{align}
where the first term in (4) accounts for the encoding of the central symbol. With probability $\frac{\binom{n-1}{d-1}}{\binom{n}{d}}$, the central symbol may have been deleted. In this case, the problem reduces to the $d-1$ deletions synchronization scenario, since we can simply insert this central symbol back to the central position after synchronizing the remaining $d-1$ deletions. This also explains the second term in (4). The third term in (4) follows from that fact that if the central symbol is successfully matched in $\sigma^{Y}$, with probability $\frac{1}{2^d}$ there are $j$ deletions in the left half of $\sigma^{X}$ and $d-j$ deletions in the right half of $\sigma^{X}$, where $0\le j\le d$. This holds since all $\binom{n}{d}$ deletion patterns are equally likely.  Inequality (5) is true because $N_c(d, n)$ is decreasing in $n$. From (6), we get
\begin{align}
[1-\frac{n-d}{n2^{d-1}}]\mathbb{E}[N_c(d,n)]&\le\log n+\frac{d}{n}\mathbb{E}[N_c(d-1,n)]\\
&\quad+\frac{n-d}{n2^d}\sum_{j=1}^{d-1}\binom{d}{j}\big{(}\mathbb{E}[N_{c}(j,n)]+\mathbb{E}[N_{c}(d-j,n)]\big{)}\nonumber\\
&\le \log n+\frac{d}{n}2(d-2)\log n+\frac{n-d}{n2^d}\sum_{j=1}^{d-1} \binom{d}{j}2(d-2)\log n,
\end{align}
where (8) follows from the induction hypothesis.  For $d\le 2$, we have
\begin{align}
\mathbb{E}[N_c(d,n)]&\le \frac{1+\frac{d}{n}2(d-2)+\frac{n-d}{n2^d}2(d-2)\sum_{j=1}^{d-1}\binom{d}{j}}{1-\frac{n-d}{n}2^{-(d-1)}}\log n\nonumber\\
&\le 2(d-1)\log n.\nonumber
\end{align}
Denote the number of anchors that have no match in $\sigma^{Y}$ by $M$, and the lengths of substrings $\sigma^{X}$ that contain single deletion errors by $l_1, \cdots, l_{d-M}$. The transmitter needs to send the $VT$--syndromes and encoding of the sums $CS_j$, where $j=1, \cdots, d-M$, for each of $d-M$ substrings that contain a single deletion. Note that the $l_j$'s and $CS_j$'s are correlated random variables. We hence have

\begin{align}
\mathbb{E}[N_v+N_s]&=\mathbb{E}[\sum_{j=1}^{d-M}\log (l_j+1)+\sum_{j=1}^{d-M}\log CS_j]\\
&\le\mathbb{E}[\sum_{j=1}^{d}\log (l_j+1)+\sum_{j=1}^{d}\log CS_j]\\
&\le\mathbb{E}[\sum_{j=1}^{d}\log (l_j+1)+\sum_{j=1}^{d}\log \frac{CS_j}{l_j}l_j]\\
&\le \mathbb{E}[\sum_{j=1}^{d}2\log (l_j+1)+\sum_{j=1}^{d}\log \frac{CS_j}{l_j}]\\
&\overset{(a)}{\le} 2d\mathbb{E}[\log \frac{\sum_{j=1}^{d}(l_j+1)}{d}]+d\mathbb{E}[\log \frac{\sum_{j=1}^{d}(\frac{CS_j}{l_j})}{d}]\\
&\overset{(b)}{\le} 2d\log \mathbb{E}[\frac{\sum_{j=1}^{d}(l_j+1)}{d}]+d\log \mathbb{E}[\frac{\sum_{j=1}^{d}(\frac{CS_j}{l_j})}{d}],
\end{align}
where $(a)$ is a consequence of  the concavity of $\log$ and $(b)$ follows from Jensen's inequality. In addition, it is easy to see $\sum_{j=1}^{d}(l_j+1)\le n$. 

\begin{align}
\mathbb{E}[\frac{\sum_{j=1}^{d}(\frac{CS_j}{l_j})}{d}]&=\frac{1}{d}(\sum_{j=1}^{d}(\mathbb{E}[\frac{CS_j}{l_j}]))\\
&=\frac{1}{d}\sum_{j=1}^{d}(\frac{\sum_{i=1}^{l_j}\mathbb{E}[\sigma^{X}_{j_i}]}{l_j})\\
&=\mathbb{E}[\sigma^{X}_1]=\frac{n+1}{2}.
\end{align}


Therefore,
\[
\begin{split}
\mathbb{E}[N_{T\to R}(d,n)]&=\mathbb{E}[N_{c}(d,n)]+\mathbb{E}[N_{v}(d)]+\mathbb{E}[N_{s}(d)]\\
&\le 2(d-1)\log n+2d\log \frac{n}{d}+d\log \frac{n+1}{2}\\
&=(5d-2)\log n-2d\log d-d\log 2+o(1).
\end{split}
\]
Let $\sigma^{X}_l(i)$ and $\sigma^{X}_r(i)$ be the left half and the right half of $\sigma^{X}(i)$, respectively. Denote the VT-syndromes of the left and right half of the substrings by $VT_l$ and $VT_r$, respectively, and use a similar notation for the checksums of the substrings, namely $CS_l$ and $CS_r$.
On the feedback link, the receiver sends out at each round the encoding of one of the nine messages:\\
(1)``failed to find a match";\\
(2)``parse $\sigma^{X}_l(i)$ and $\sigma^{X}_r(i)$";\\
(3)``parse $\sigma^{X}_l(i)$ and send $VT_r$";\\
(4)``parse $\sigma^{X}_l(i)$ and send $CS_r$";\\
(5)``send $VT_l$ and parse $\sigma^{X}_r(i)$";\\
(6)``send $VT_l$ and $VT_r$";\\
(7)``send $VT_l$ and $CS_r$";\\
(8)``send $CS_l$ and parse $\sigma^{X}_r(i)$";\\
(9)``send $CS_l$ and $VT_r$".\\
The number of bits transmitted by the receiver is at most three bits at each round. Therefore,
\begin{align}
\mathbb{E}[N_{R\to T}(d)]&\le 3 \frac{\mathbb{E}[N_c(d)]}{\log n}=6(d-1).
\end{align}
\end{proof}

%

For the case of insertion errors, the situation is reversed in so far that the transmitter is in possession of a partial permutation, while the receiver contains a permutation. Interestingly, one only needs to identify the inserted symbols, since their positions are automatically revealed thereafter. This reduces the total number of transmitted bits by $d\log n$.

\subsection{Block deletions/insertions}\label{sec: burstdeletions}
We consider next the problem of synchronizing from block deletions. Since deletions occur in consecutive order, the receiver only needs to know the first or the last edited position, as well as the arrangement of the $d$ deleted symbols. In the genie-aided case, the required number of transmitted bits equals
\[
\log (n-d+1) +\log d!=\log n+d\log d+O(d).
\]

Clearly, the deletion synchronization method described in the previous section also applies to the block deletion case.
However, the communication throughput for the random deletion protocol may be significantly higher than needed, given that the deletions appear in consecutive positions. To see this, consider an example with $d=2$. On average, the random synchronization protocol communicates $O(\log^2n)$ bits and $O(\log n)$ bits through the forward link and the feedback link, respectively. The protocol we propose next only requires a $O(\log\, d\, \log\,n)$ throughput on the forward link.

We start by introducing the process of deinterleaving.
In the deinterleaving process, $\sigma^{X}$ and $\sigma^{Y}$ are parsed into $d$ subsequences $(\sigma^{X})^k$ and $ (\sigma^{Y})^k$ of the form
\[
\begin{split}
(\sigma^{X})^k=(\sigma^{X}_k, \sigma^{X}_{k+d}, \sigma^{X}_{k+2d},\cdots);\\
(\sigma^{Y})^k=(\sigma^{Y}_k, \sigma^{Y}_{k+d}, \sigma^{Y}_{k+2d},\cdots),
\end{split}
\]
where, for $k=1,\cdots,d$,  $(\sigma^{X})^k$ and $(\sigma^{Y})^k$ are mis-synchronized by one deletion only. For instance, suppose that
the transmitter stores
\[
\sigma^{X}=( 1, 14, 12, 2, \textcolor{red}{3, 4, 9}, 10, 11, 13, 5, 8, 7, 6, 15 ),
\] while the noisy version available at the receiver reads as
\[
\sigma^{Y}=(1, 14, 12, 2, 10, 11, 13, 5, 8, 7, 6, 15).
\]
The above described parsing method results in:
\[
\begin{split}
&(\sigma^{X})^1, (\sigma^{Y})^1=(1, 2, \textcolor{red}{9}, 13, 7), (1, 2, 13, 7),\\
&(\sigma^{X})^2, (\sigma^{Y})^2=(14, \textcolor{red}{3}, 10, 5, 6), (14, 10, 5, 6),\\
&(\sigma^{X})^3, (\sigma^{Y})^3=(12, \textcolor{red}{4}, 11, 8, 15), (12, 11, 8, 15).
\end{split}
\]

 The resulting ``single" deletion synchronization can be done via \emph{one-way communication} by letting the transmitter send out the VT-syndromes and checksums for each of the $d$ substrings $(\sigma^{X})^k$, for $1\le k\le n$. The total number of transmitted bits is
\[
N_{T\to R}(d)+N_{R\to T}(d)=N_{T\to R}(d)=3d\log n.
\]
As presented in Protocol 3, the total communication throughput can be improved to $O(\log d\log n)$, with $O(\log d)$ bits transmitted on the feedback link.
The key idea is to utilize the error structure. Denote the position of the symbol deleted in $(\sigma^{X})^i$  by $p_i$. If a deletion in $(\sigma^{X})^1$ occurred at position $j$, i.e., if $p_1=j$, then for $i\ge 2$, $p_i$ equals either $j$ or $j-1$, which is a consequence of the fact that deletions occur in consecutive order.
In particular, the sequence $\{p_i\}_{i=1}^d$ equals
 \begin{align}
 (p_1,\cdots, p_{k-1}, p_{k},\cdots, p_d)=(j, \cdots, j, j-1, \cdots, j-1)\nonumber,
 \end{align}
where $k$ denotes the index of the subsequence of $\sigma^{X}$ containing the first deleted symbol. Note that we may have $k=d+1$, implying that the first deleted symbol is contained in $(\sigma^{X})^1$.
It is straightforward to see that the first deleted position $p^*$ equals
\[
p^*=(j-1)d+1, \; \text{if} \; p_1=p_d, \; \text{and}
\]
\[
p^*=(j-1)d+k-d=(j-2)d+k, \; \text{otherwise}.
\]
\begin{algorithm}
\caption{Block Deletion Protocol}
\BlankLine
Initialization: $m\gets0$, $t\gets0$ and $\mathcal{I}^{(0)}\gets\{i_1,\cdots, i_d\}=\{1, \cdots, d\}$\;
Transmitter sends the VT-syndromes and the checksums for $(\sigma^{X})^1$ and $(\sigma^{X})^d$\;
Receiver recovers $(\sigma^{X})^1$ from $(\sigma^{Y})^1$ and computes $p_1$\;
Receiver recovers $(\sigma^{X})^d$ from $(\sigma^{Y})^d$ and computes $p_d$\;
\eIf{$p_1=p_d$}
{Receiver sends ``FOUND" and $k=1$\;}
{\While {$\mathcal{I}^{(t)}$ is not singleton}
{$m\gets \lceil\frac{|\mathcal{I}^{(t)}|}{2}\rceil$\;
Transmitter sends the VT-syndrome and the checksum of  $(\sigma^{X})^m$\;
Receiver recovers $(\sigma^{X})^m$ from $(\sigma^{Y})^m$ and computes $p_m$\;
\eIf{$p_1>p_m=p_d$}
{$\mathcal{I}^{(t+1)}\gets \{i^{(t)}_1,\cdots, i^{(t)}_{m}\}$}
{$\mathcal{I}^{(t+1)}\gets \{i^{(t)}_{m},\cdots, i^{(t)}_{|\mathcal{I}^{(t)}|}\}$}
$t\gets t+1$\;
Receiver sends ``NOT FOUND"\;}
Receiver sends ``FOUND" and\\
\eIf{$p_1>p_m=p_d$}
{sends $k=m$\;}
{sends $k=m+1$.}
}
\label{burstprotocol}
\end{algorithm}

\begin{theorem}
Protocol 3 exactly restores $\sigma^{X}$ at the receiver, with
\begin{align}
\mathbb{E}[N_{T\to R}(d)]&= 3\log d \, \log n+6\log n+\log d!-\frac{2\log d}{d}\nonumber,\\
var[N_{T\to R}(d)]&=\frac{9(d-1)}{d^2}\log^2d \, \log^2n\nonumber,
\end{align}
and 
\begin{align}
\mathbb{E}[N_{R\to T}(d)]&=\frac{2d-1}{d}\log d\nonumber,\\
var[N_{R\to T}(d)]&=\frac{d-1}{d^2}\log^2d.
\end{align}
\end{theorem}

\begin{proof}

When $(\sigma^{X})^1$ and $(\sigma^{Y})^1$ are synchronized, the receiver knows that $(p_1-1)d+1$ is in the span of the block deletion. Since all $n-d+1$ block deletions patterns may have occurred equally likely, with probability $\frac{1}{d}$, $(p_1-1)d+1$ is the first edited position, which can be detected by the receiver via comparing $p_1$ with $p_d$. In this case, Protocol 3 terminates with step 6, and $N_{T\to R}(d)=6\log n$, $N_{R\to T}(d)=\log d$. Otherwise, the protocol goes though the \textbf{while} command, which terminates at round $\log d$ when $\mathcal{I}^{(\log d)}$ is a singleton. In the latter case, we have
\[
N_{T\to R}(d)=3\log d \log n+6\log n
\]
and
\[
N_{R\to T}(d)=\frac{d-1}{d}2\log d.
\]
Then
\begin{align}
\mathbb{E}[N_{T\to R}(d)]&=\frac{1}{d}6(\log n+\log d!)+\frac{d-1}{d}(6\log n+3\log d\log n+\log d!)\nonumber\\
&=3\log d \, \log n+6\log n+\log d!-\frac{2\log d}{d}\nonumber
\end{align}
The expressions $var[N_{T\to R}(d)]$ and $var[N_{R\to T}(d)]$ may be derived similarly.
\end{proof}

\section{A Single Translocation: A Pair of a Deletion and an Insertion}\label{sec: translocation}
On a permutation of length $n$, one can perform as many as $(n-1)^2$ different translocations. Thus, in the genie-aided case, $2\log (n-1)=2\log n+o(1)$ bits need to be transmitted. We describe next a protocol that is within factor of three from the genie-aided limit.

First, observe that a single translocation error is equivalent to a deletion and an insertion of the same symbol~\cite{FarnoudSM13}. Hence, the idea is to partition $\sigma^{X}$ in such a way that the deletion error and the insertion error are contained in different substrings of $\sigma^{X}$. To correct the transposition, we use the fact that VT-codes for permutations are capable of \emph{detecting} single translocations.

Let $S_{\sigma^{X}}$ and $S_{\sigma^{Y}}$ be the to-be-parsed substrings of $\sigma^{X}$ and $\sigma^{Y}$, respectively.
\vspace{-0.15in}
\begin{algorithm}
\caption{Protocol for Single Translocation}
\BlankLine
Initialization: $S_{\sigma^{X}}\gets \sigma^{X}$, $S_{\sigma^{Y}}\gets \sigma^{Y}$\;
Transmitter sends the central symbol of $S_{\sigma^{X}}$\;
Receiver anchors the central symbol in $S_{\sigma^{Y}}$\;
\If{the central symbol was not shifted}
{The receiver requests $VT_l(S_{\sigma^{X}})$\;
\eIf{$VT_l(S_{\sigma^{X}})\not=VT_l(S_{\sigma^{Y}})$}
{$S_{\sigma^{X}}\gets \sigma^{X}_l$, $S_{\sigma^{Y}}\gets \sigma^{Y}_l$, go to step 2\;}
{$S_{\sigma^{X}}\gets \sigma^{X}_r$, $S_{\sigma^{Y}}\gets \sigma^{Y}_r$, go to step 2\;}
}
\eIf{the central symbol was shifted by one position to the left}
{The receiver requests $CS_r(S_{\sigma^{X}})$ and $VT_l(S_{\sigma^{X}})$, uses $CS_r(S_{\sigma^{X}})$ to synchronize the insertion in the right part of $S_{\sigma^{Y}}$ and uses $VT_l(S_{\sigma^{X}})$ to synchronize the deletion in the left part of $S_{\sigma^{Y}}$\;}
{The receiver requests $CS_l(S_{\sigma^{X}})$ and $VT_r(S_{\sigma^{X}})$, uses $CS_l(S_{\sigma^{X}})$ to synchronize the insertion in the left part of $S_{\sigma^{Y}}$ and uses $VT_r(S_{\sigma^{X}})$ to synchronize the deletion in the right part of $S_{\sigma^{Y}}$.}
\label{translocation}
\end{algorithm}

The protocol starts with the transmitter sending the central symbol of $\sigma^{X}$, i.e., the symbol at position $\lceil \frac{n}{2}\rceil$ in $\sigma^{X}$, to the receiver. The receiver examines whether the position of the received symbol is $\lceil \frac{n}{2}\rceil$ in $\sigma^{Y}$. If not, the received symbol is within the span of the translocation, and a deletion occurred in the left half of $\sigma^{X}$, and an insertion occurred in the right half of $\sigma^{X}$, or vice versa. If the received symbol is accurately anchored at $\lceil \frac{n}{2}\rceil$, the protocol uses the VT-syndrome to determine which half of $\sigma^{X}$ contains the translocation. The process is repeated for the substring that contains the translocation error.
\begin{theorem}
Protocol 4 exactly restores $\sigma^{X}$ at the receiver, with the number of bits transmitted through the forward link satisfying
\begin{align}
\mathbb{E}[N_{T\to R}]&\le 6\log n,\\
var[N_{T\to R}]&\le 8\log^2 n+O\big{(}\frac{\log^2n}{n}\big{)},
\end{align}
and the number of bits transmitted through the feedback link satisfying
\begin{align}
\mathbb{E}[N_{R\to T}]&\le 6,\\
var[N_{R\to T}]&\le 18+O\big{(}\frac{1}{n}\big{)}.
\end{align}
\label{thmtranslocation}
\end{theorem}
\begin{remark}
Due to the symmetry of a translocation, the limited feedback protocol can be easily adapted for a forward link limited model by exchanging the roles of the transmitter and the receiver. 
\end{remark}
\begin{proof}
Let $M$ be the random variable counting the transmission rounds needed for Protocol 4 to terminate. Denote the distribution of $M$ by $\mathbb{Q}_{M}$.

If Protocol 4 terminates at round $M=m$, by that point, the transmitter has sent $m$ anchor symbols, $m-1$ VT-syndromes for detecting the translocation within the first $m-1$ rounds, and $2\log n$ bits and $\log n$ bits for synchronizing first from the insertion and then deletion error, respectively. Hence, the total number of bits sent by the transmitter equals $(2m+2)\log n$, and $\mathbb{E}[N_{T\to R}]$ and $var[N_{T\to R}]$ may be written as
\begin{align}
\mathbb{E}[N_{T\to R}]&=2\log n \,\mathbb{E}[M]+2\log n,\\
var[N_{T\to R}]&=(4\log^2n) \, var[M].
\end{align}
On the feedback link, the receiver sends out at each round the encoding of one of the five messages: (1) ``send $VT_l(S_{\sigma^{X}})$''; (2) ``parse $S_{\sigma^{X}}\gets \sigma^{X}_l$, $S_{\sigma^{Y}}\gets \sigma^{Y}_l$''; (3) ``parse $S_{\sigma^{X}}\gets \sigma^{X}_r$, $S_{\sigma^{Y}}\gets \sigma^{Y}_r$''; (4) ``send $CS_r(S_{\sigma^{X}})$ and $VT_l(S_{\sigma^{X}})$''; (5) ``send $CS_l(S_{\sigma^{X}})$ and $VT_r(S_{\sigma^{X}})$''. For the encoding, only three bits are needed.
Thus, we have
\begin{align}
\mathbb{E}[N_{R\to T}]&= 3 \, \mathbb{E}[M],\\
var[N_{R\to T}]&=9 \, var[M].
\end{align}
Next, we bound the moments $\mathbb{E}[M]$ and $var[M]$.

Protocol \ref{translocation} terminates at round $m$ if and only if the anchor symbol sent at round $m$ was shifted in $\sigma^{Y}$. Denote the probability of the event ``the $k^{th}$ entry in $\sigma^{X}$ was shifted in $\sigma^{Y}$" by $\mathbb{P}_k$. If the $k^{th}$ entry in $\sigma^{X}$ was not shifted in $\sigma^{Y}$, then either the translocation error was contained within the first $k-1$ positions or contained within the last $n-k$ positions. For permutation strings of length $k-2$ and $n-k$, one can perform $(k-2)^2$ and $(n-k-1)^2$ different translocations, respectively. Thus we have,
\begin{align}
\mathbb{P}_k=1-\frac{(k-2)^2+(n-k-1)^2}{(n-1)^2},
\end{align}
which is maximized at $k^{*}=\lceil\frac{n}{2}\rceil$. Since in the first round the center symbol $\sigma^{X}_{\lceil\frac{n}{2}\rceil}$ is checked, the probability that the protocol terminates at round one is
\begin{align}
\mathbb{Q}_{M}(M=1)=\mathbb{P}_{k^{*}}=
\begin{cases}
\frac{1}{2}+\frac{2}{n-1}-\frac{2}{(n-1)^2}& \text{if}\,\, n \,\text{\,is \,odd};\\
\frac{1}{2}+\frac{2}{n-1}-\frac{5}{2(n-1)^2}& \text{otherwise}.
\end{cases}
\end{align}
If the received symbol is accurately anchored at $\lceil \frac{n}{2}\rceil$, the protocol uses the VT-syndrome to determine which half of $\sigma^{X}$ contains the translocation. The process is repeated for the substring that contains the translocation error.
The length of the substring of interest at each round is characterized in Lemma \ref{lemmatranslocation}

\begin{lemma}
Let $\{a_k\}_{k=1}^{\infty}$ be a sequence such that $a_i$ denotes the length of the substring of $\sigma^{X}$ (or, equivalently, $\sigma^{Y}$) at round $k$. Then
\begin{align}
a_k=
\begin{cases}
\frac{n+1-2^{k-1}}{2^{k-1}} &\forall\, k\le \log (n+1)-1\\
0& \text{otherwise}.
\end{cases}
\end{align}
\label{lemmatranslocation}
\begin{proof}
We prove this claim by induction.

Base Case: When $k=1$, $\frac{n+1-2^{1-1}}{2^{1-1}}=n=a_1$.

Induction Hypothesis: Suppose that $a_k=\frac{n+1-2^{k-1}}{2^{k-1}}$ for all $k\le\log (n+1)-2$.

It is straightforward to see that
\begin{align}
a_{k+1}&=\frac{\frac{n+1-2^{k-1}}{2^{k-1}}+1}{2}-1.\\
&=\frac{n+1-2^k}{2^k}.
\end{align}
The algorithm performs splitting until the substrings reach a threshold length which cannot be smaller than three. Hence
\begin{align}
&a_k=\frac{n+1-2^{k-1}}{2^{k-1}}\ge 3\\
&\Rightarrow k\le \log(n+1)-1.
\end{align}
Since at most $\log(n+1)-1$ rounds are needed, $a_k=0$ for all $k\ge \log(n+1)$.
\end{proof}
\end{lemma}
As a result, the distribution of $M$ has the following closed form
\begin{align}
\mathbb{Q}_{M}(m)&=\left(\frac{1}{2}+\frac{2}{\frac{n+1-2^{m-1}}{2^{m-1}}-1}-\frac{2}{(\frac{n+1-2^{m-1}}{2^{m-1}}-1)^2}\right)\times \nonumber\\ &\prod_{i=1}^{m-1}\left(\frac{1}{2}-\frac{2}{\frac{n+1-2^{i-1}}{2^{i-1}}-1}+\frac{2}{(\frac{n+1-2^{i-1}}{2^{i-1}}-1)^2}\right),\nonumber
\end{align}
for $m\le \log(n+1)-1$; and $\mathbb{Q}_{M}(m)=0$ otherwise.
%

Suppose next that $G$ is a geometric random variable with parameter $\frac{1}{2}$. It can be shown by induction that the random variable $M$ is first-order stochastically dominated by $G$, i.e., for all $m$,
\begin{align}
\mathbb{Q}_{M}(M\le m)>\mathbb{P}(G\le m),
\end{align}
%
which immediately implies
\[
\mathbb{E}[M]\le\mathbb{E}[G]=2.
\]
Nevertheless, the claim that $var[M]\le var[G]$ may not hold in general. 
Still, we may write
\begin{align}
var[M]&=\mathbb{E}[M-\mathbb{E}[M]]^2\nonumber\\
&=\mathbb{E}_{M\le \mathbb{E}[M]}[M-\mathbb{E}[M]]^2+\mathbb{E}_{M\ge \mathbb{E}[M]}[M-\mathbb{E}[M]]^2.
\label{stochasticdominance}
\end{align}
By observing that $1<\mathbb{E}[M]<2$, the first term on the right hand side of (\ref{stochasticdominance}) can be bounded as
\begin{align}
\mathbb{E}_{M\le\mathbb{E}[M]}[M-\mathbb{E}[M]]^2\le \mathbb{Q}(1)=\frac{1}{2}+O(\frac{1}{n})\nonumber.
\end{align}
Similarly, it can be shown that the second term on the right hand side of (\ref{stochasticdominance}) satisfies
\begin{align}
\mathbb{E}_{M\ge \mathbb{E}[M]}[M-\mathbb{E}[M]]^2\le \mathbb{E}_{G\ge \mathbb{E}[G]}[G-\mathbb{E}[G]]^2\nonumber,
\end{align}
which completes the proof.
\end{proof}

\section{Synchronization from a Single Transposition Error}\label{sec: transposition}
Suppose that $\sigma^{Y}=\sigma^{X}\tau$, where $\tau$ is a transposition. Let $\tau=(a\,b)$, where $a,b\in [n]$ and $a<b$, implying that the elements $\sigma^{X}_a$ and $\sigma^{X}_b$ were swapped. In this scenario, the genie-aided lower bound equals $\log \binom{n}{2}=2\log n+O(1)$.

We first show that anchoring strategies cannot lead to order optimal protocols. Since a transposition is equivalent to two substitution errors, an anchoring strategy reduces to a trivial ``send and check" interaction, i.e., the transmitter keeps sending different symbols until one of the swapped symbols is identified. Denote the number of rounds before the protocol terminates by $M_{\tau}$. Since
\begin{align}
\mathbb{E}[M_{\tau}]&=\sum_{k=1}^{n}\mathbb{P}[M_{\tau}\ge k]=\frac{n+1}{3},\nonumber\\
\text{where}\,\,~
\mathbb{P}[M_{\tau}\ge k]&=\frac{\binom{n-2}{k-1}}{\binom{n}{k-1}}=\frac{(n-k)(n-k+1)}{n(n-1)},
\end{align}
the average number of transmitted bits equals
\[
\mathbb{E}[N_{T\to R}]=\mathbb{E}[M_{\tau}]\log n=\frac{n+1}{3}\log n.
\]
We show next that a single transposition can be synchronized using an one-way protocol in which the transmitter sends the encoding of three quantities: $\delta^{X}_1=\Sigma_{i=1}^{n}i\,\sigma^{X}_i$, $\delta^{X}_2=\Sigma_{i=1}^{n}i^2 \, \sigma^{X}_i$ and $\delta^{X}_3=\Sigma_{i=1}^{n}i^3 \, \sigma^{Y}_i$. Similarly, let $\delta^{Y}_1=\Sigma_{i=1}^{n} i \, \sigma^{Y}_i$, $\delta^{Y}_2=\Sigma_{i=1}^{n}i^2 \, \sigma^{Y}_i$ and $\delta^{Y}_3=\Sigma_{i=1}^{n}i^3 \,\sigma^{Y}_i$. The receiver computes $a$ and $b$ from
\begin{align}
\begin{cases}
\delta^{Y}_1-\delta^{X}_1&=(\sigma^{X}_{b}-\sigma^{X}_{a})(a-b);\\
\delta^{Y}_2-\delta^{X}_2&=(\sigma^{X}_{b}-\sigma^{X}_{a})(a-b)(a+b);\\
\delta^{Y}_3-\delta^{X}_3&=(\sigma^{X}_{b}-\sigma^{X}_{a})(a-b)(a^2+b^2+a\,b).\nonumber
\end{cases}
\end{align}
and then solves the system of equations
\begin{align}
a+b=\frac{\delta^{Y}_2-\delta^{X}_2}{\delta^{Y}_1-\delta^{X}_1}; \;\; \;
a^2+b^2+a \, b=\frac{\delta^{Y}_3-\delta^{X}_3}{\delta^{Y}_1-\delta^{X}_1}.
\end{align}
The average number of transmitted bits equals $12\log n$.

Note that the moment sums $\delta^{X}_i$, $i=1,2,3$, may be seen as generalized VT-syndromes as well as ordinal Reed-Solomon type parity-checks.


\section{Conclusion}

In this work, we have explored the problem of synchronizing ordinal data with special attention to the scenario when there is stringent constraint on the feedback link throughput per synchronization procedure. Four types of information edits--random deletions/insertions, block deletions/insertions, single transloations and single transpositions--have been analyzed individually.
For $\sigma^{Y}$ and $\sigma^{X}$ mis-synchronized by deletions, we exhibit protocols within a factor of two and a factor of five from the genie-aided limits for $c_{tr} \simeq c_{rt}$ and $c_{tr}\gg c_{rt}$, respectively. When the synchronization error is a single translocation, a protocol within a factor of three from the genie-aided limit is proposed. For single transposition errors, we describe a one-way protocol within a factor of six from the genie-aided limit. This protocol uses generalization of Varshamov-Tenengolz and Reed-Solomon codes for ordinal information.

\bibliographystyle{plain}
\bibliography{synbib}


%

\appendices
%

\section*{Acknowledgment}

This work was supported in part by NSF grants CCF 0809895, CCF 1218764 and the Emerging Frontiers for Science of Information Center, CCF 0939370.

\ifCLASSOPTIONcaptionsoff
  \newpage
\fi

\end{document}